\documentclass[copyright,creativecommons]{eptcs}
\newcommand{\typeof}{1}
\providecommand{\event}{DICE 2010}
\providecommand{\volume}{??}
\providecommand{\anno}{20??}
\providecommand{\firstpage}{1}
\providecommand{\eid}{??}

\usepackage{ifthen}
\newcommand{\condinc}[2]{\ifthenelse{\equal{\typeof}{0}}{#1}{#2}}

\usepackage{proof}
\usepackage{amsmath}
\usepackage{subfigure}
\usepackage{epsf}
\usepackage{graphics}
\usepackage{wrapfig}
\usepackage{mathrsfs}
\usepackage{url}
\usepackage{calc}
\usepackage{amsmath,amssymb,amsfonts,mathrsfs,latexsym,stmaryrd}
\usepackage[all]{xy}

\newcommand{\funone}{\mathbf{f}}
\newcommand{\funtwo}{\mathbf{g}}
\newcommand{\funthree}{\mathbf{h}}
\newcommand{\funfour}{\mathbf{p}}
\newcommand{\tfunone}{\mathsf{f}}
\newcommand{\tfuntwo}{\mathsf{g}}
\newcommand{\tfunthree}{\mathsf{h}}

\newcommand{\funsetone}{S}
\newcommand{\funsettwo}{T}
\newcommand{\tiersone}{\mathsf{I}}
\newcommand{\tierstwo}{\mathsf{J}}
\newcommand{\tiersthree}{\mathsf{K}}
\newcommand{\arityone}{\sigma}
\newcommand{\aritytwo}{\tau}
\newcommand{\algone}{\mathcal{A}}

\newcommand{\terms}[1]{\mathcal{T}_{#1}}
\newcommand{\conone}{\mathbf{c}}

\newcommand{\termone}{t}
\newcommand{\termtwo}{u}

\newcommand{\contone}{C}

\newcommand{\sigone}{\Sigma}
\newcommand{\sigtwo}{\Theta}
\newcommand{\sigthree}{\Xi}
\newcommand{\vsone}{V}
\newcommand{\vstwo}{W}
\newcommand{\vsthree}{X}
\newcommand{\vsfour}{Y}
\newcommand{\ordone}{\alpha}
\newcommand{\ordtwo}{\beta}
\newcommand{\ordthree}{\gamma}
\newcommand{\labelone}{\delta}
\newcommand{\labeltwo}{\epsilon}
\newcommand{\labelthree}{\zeta}
\newcommand{\subone}{\xi}
\newcommand{\subtwo}{\theta}
\newcommand{\rootone}{r}
\newcommand{\roottwo}{s}
\newcommand{\rootthree}{p}
\newcommand{\rootfour}{q}
\newcommand{\verone}{v}
\newcommand{\vertwo}{w}
\newcommand{\verthree}{x}
\newcommand{\verfour}{y}
\newcommand{\tgone}{G}
\newcommand{\tgtwo}{H}
\newcommand{\tgthree}{J}
\newcommand{\tgfour}{K}
\newcommand{\tgfive}{I}
\newcommand{\tgsix}{L}
\newcommand{\tgseven}{M}
\newcommand{\rrone}{\rho}
\newcommand{\rrtwo}{\sigma}
\newcommand{\homone}{\varphi}
\newcommand{\homtwo}{\psi}

\newcommand{\rewrite}[1]{\stackrel{#1}{\longrightarrow}}
\newcommand{\rewritecbv}[1]{\stackrel{#1}{\longrightarrow_\mathsf{i}}}
\newcommand{\rewritecbn}[1]{\stackrel{#1}{\longrightarrow_\mathsf{o}}}
\newcommand{\id}{\mathsf{id}}
\newcommand{\funid}{\mathbf{id}}
\newcommand{\proj}[2]{\mathsf{p}_{#1,#2}}
\newcommand{\tproj}[2]{\mathbf{p}_{#1,#2}}
\newcommand{\comp}[1]{\mathsf{comp}(#1)}
\newcommand{\rec}[1]{\mathsf{rec}(#1)}
\newcommand{\cond}[1]{\mathsf{cond}(#1)}
\newcommand{\subst}[3]{#1_{#2}(#3)}

\newcommand{\tiered}[1]{\mathbb{N}#1}
\newcommand{\ptiered}[2]{#2#1}

\newcommand{\constr}[1]{\mathsf{s}_{#1}}
\newcommand{\funconstr}[2]{\mathbf{s}_{#1}^{#2}}

\newcommand{\size}[1]{|#1|}
\newcommand{\psize}[2]{|#1|_{#2}}
\newcommand{\subsig}{\sqsubseteq}
\newcommand{\supsig}{\sqsupseteq}
\newcommand{\var}[1]{\mathit{Var}(#1)}

\newcommand{\qed}{\hfill$\Box$}

\newcommand{\GtoT}[1]{\langle #1\rangle}

\newcommand{\domain}[1]{\mathit{dom}(#1)}

\newcommand{\sgrone}{\mathcal{G}}
\newcommand{\sgrtwo}{\mathcal{H}}

\newcommand{\rewrgraph}{\rightarrow}
\newcommand{\rewrgraphcbv}{\rightarrow_\mathsf{i}}
\newcommand{\rewrgraphcbn}{\rightarrow_\mathsf{o}}

\newcommand{\subgr}[2]{#1\downarrow #2}

\condinc{\newcommand{\N}{\mathbb{N}}}{\newcommand{\N}{\mathbb{N}}}

\newtheorem{lemma}{Lemma}
\newtheorem{proposition}{Proposition}
\newtheorem{theorem}{Theorem}

\newenvironment{proof}{\begin{trivlist}
       \item[\hskip \labelsep {\bfseries Proof.}]}{\hfill $\Box$ \end{trivlist}}
\newtheorem{definition}{Definition}
\newtheorem{example}{Example}

\condinc{
\newenvironment{varitemize}
{
\begin{list}{\labelitemi}
{\setlength{\itemsep}{0.0mm}
 \setlength{\topsep}{0.0mm}
 \setlength{\parindent}{0.0mm}
 \setlength{\parskip}{0.0mm}
 \setlength{\parsep}{0.0mm}
 \setlength{\partopsep}{0.0mm}
 \setlength{\leftmargin}{15pt}
 \setlength{\labelsep}{5pt}
 \setlength{\labelwidth}{10pt}}}
{
 \end{list} 
}}{
\newenvironment{varitemize}
{
\begin{list}{\labelitemi}
{\setlength{\itemsep}{0.0mm}
 \setlength{\topsep}{0.0mm}
 \setlength{\parindent}{0.0mm}
 \setlength{\parskip}{0.0mm}
 \setlength{\parsep}{0.0mm}
 \setlength{\partopsep}{0.0mm}
 \setlength{\leftmargin}{15pt}
 \setlength{\labelsep}{5pt}
 \setlength{\labelwidth}{10pt}}}
{
 \end{list} 
}}

\newcounter{number}

\newenvironment{varenumerate}
{\begin{list}{\arabic{number}.}
  {
   \usecounter{number}
   \setlength{\labelwidth}{4.0mm}
   \setlength{\labelsep}{2.0mm}
   \setlength{\itemindent}{0.0mm}
   \setlength{\itemsep}{0.0mm}
   \setlength{\topsep}{0.0mm}
   \setlength{\parskip}{0.0mm}
   \setlength{\parsep}{0.0mm}
   \setlength{\partopsep}{0.0mm}
  }
}
{\end{list}}

\title{General Ramified Recurrence\\ is Sound for Polynomial Time}
\author{{Ugo Dal Lago\footnote{
Dipartimento di Scienze dell'Informazione, Universit\`a di Bologna, 
Mura Anteo Zamboni 7, 40127 Bologna, Italy.
\texttt{dallago@cs.unibo.it}
}
}
\and 
{{Simone Martini}\footnote{
Dipartimento di Scienze dell'Informazione, Universit\`a di Bologna, 
Mura Anteo Zamboni 7, 40127 Bologna, Italy.
\texttt{martini@cs.unibo.it}
}}
\and
{{Margherita Zorzi}\footnote{
Dipartimento di Informatica, Universit\`a di Verona, 
Strada le Grazie 15, 37134 Verona, Italy.
\texttt{margherita.zorzi@univr.it}
}}}

\def\titlerunning{General Ramified Recurrence is Sound for Polynomial Time}
\def\authorrunning{U. Dal Lago \& S. Martini \& M. Zorzi}
\begin{document}
\maketitle
\begin{abstract}
\noindent
Leivant's ramified recurrence is one of the earliest examples of an implicit characterization of the polytime 
functions as a subalgebra of the primitive recursive functions. Leivant's result, however, is originally stated 
and proved only for word algebras, i.e.\ free algebras whose constructors take at most one argument. This paper 
presents an extension of these results to ramified functions on any free algebras, provided the underlying 
terms are represented as graphs rather than trees, so that sharing of identical subterms can be exploited.
\end{abstract}
\section{Introduction}
The characterization of complexity classes by language restrictions
(i.e., by \emph{implicit} means) instead of explicit resource bounds
is a major accomplishment of the area at the intersection of logic and
computer science. Bellantoni, Cook~\cite{BC92}, and
Leivant~\cite{Lei93}, building on Cobham pioneering
research~\cite{Cobham}, gave two (equivalent) restrictions on the
definition of the primitive recursive functions, obtaining in this way
exactly the functions computable in polynomial time.  We will focus in
this paper on Leivant's seminal work.

There are (at least) two main ingredients in these implicit characterizations of polytime. First, when data are represented by 
strings, as usual in complexity, each recursive call must consume at least
one symbol of the input. In this way the length of the recursive call sequence is linear in the size of the input. 
When \emph{numeric} functions are considered, and numbers are thus represented in basis $b\ge 2$, this amounts to 
recursion on notation~\cite{Cobham}, where each call divides the input by $b$.
The second main ingredient is a restriction on the recursion schema, in order to avoid nested recursions. This is the job of 
\emph{tiers}~\cite{Lei93,Simmons88} (in the Bellantoni-Cook's approach this would be achieved with a distinction between safe 
and normal arguments in a function). 
In Leivant's system variables and functions are equipped with a tier, and  composition must preserve tiers; crucially, in a 
legal recursion the tier of the recurrence parameter must be higher than the tiers of the recursive calls.
It is noteworthy that \emph{linearity} does \emph{not} play a major role --- a function can duplicate its inputs as many times 
as it likes\footnote{The naive restriction to primitive recursion on notation plus linearity (and no tiers) is too generous. 
Exponential functions would be easily definable. For example the function  $\tfunone$, defined by linear recursion on notation as 
$\tfunone(\mathbf{0})=\tfunone(\mathbf{1})=\mathbf{1}$ and 
$\tfunone(w\cdot \mathbf{0})=\tfunone(w\cdot \mathbf{1})=\tfuntwo(\tfunone(w))$ 
(where $\tfuntwo$ is any recursively defined function such that $|\tfuntwo(x)|\geq 2|x|$) has superpolynomial 
growth caused by the application of $\tfuntwo$ on the result of the recursive call.}.
In Leivant’s original paper~\cite{Lei93}, ramified recurrence over any free algebra is
claimed to be computable in polynomial time on the height of the input, hence on its
size. However, some proofs (in particular, the proof of Lemma 3.8) only go through
when the involved algebras have constructors of at most unary arities. 
Indeed, the extended and revised~\cite{Lei95} only refers to word-algebras. 
Marion~\cite{Marion03} extends the polynomiality result to constructors with signature
$s_1\times\cdots\times s_n\rightarrow s$ under the constraint that $s$ appears at most once among
the $s_i$, and it is held in the ICC community that the result holds also
for \emph{any} free
algebra (see for instance Marion's observation, reported as personal communication in Hofmann's~\cite{Hof00}, page 38). 
This gap in the literature gives rise to subtle misunderstandings (which
could amount to believing the contrary: in the unpublished~\cite{Caseiro96}
we read that Leivant  ``has given equational characterizations of complexity classes, but for constructors of arity greater 
than one, his classes exceed poly-time''). In this paper we thus fill the little gap, and prove anew that Leivant's 
characterization of polytime holds for \emph{general tiered recursion},
as part of a broader project aimed to give precise complexity content to rule based programming.

The point is that Leivant's proofs does not go through when 
moving from unary to \emph{arbitrary arity} constructors, since now the absence of linearity strikes back. Indeed, the 
following function on binary words (which is easily decorated with tiers, but which violates the
constraint in the already cited~\cite{Marion03})
$$
\tfunone(\mathbf{0})=\tfunone(\mathbf{1})=\mathbf{nil}
\qquad \tfunone(w\cdot\mathbf{0})=\tfunone(w\cdot\mathbf{1})=\mathbf{tree}(\tfunone(w),\tfunone(w))
$$
outputs the full binary tree, which has exponential size in the length of the input.

We believe that this is a \emph{representation problem}, and not an intrinsic limitation of tiering.
The apparent break of polytime appears because the explicit representation of data with strings
forces the explicit duplication of (part of) the input. But this duplication is inessential to the computation 
itself --- in fact, it could be avoided by just storing the intermediate result and re-using it when needed to produce the output. 
We thus prove that tiered recursion on \emph{any free algebra} may be computed in polytime, 
once  data is represented with directed acyclic graphs, and computation is performed via 
graph rewriting. In term graph rewriting  the sharing of common subterms is explicitly represented, 
and a compact representation of  data could be given. The result of a computation will be, in general, a DAG, 
where identical subterms that would be replicated several times in the string representation, are instead present only once.
The time needed to print the string representation of the output is not (and \emph{should not be}) counted in the computing time of the function. 

The choice of a suitable representation for data is of course a crucial ingredient of any complexity theory account. Explicit string 
representation of arbitrary arity terms is simply too generous,
akin to representation of numbers in base 1. Indeed, the discriminant for an acceptable encoding
of data (e.g.,~\cite{Papa}) is the fact that all acceptable encodings yield representations
which have polynomially related lengths. And this rules out the explicit string representation, in view of the example above. 
On the other side, we think that graph representation of elements of a free algebra makes a good choice. 

The present work solves an open question about a primal approach of ICC, joining the pure functional characterization of Leivant's approach with the powerful features of graph 
rewriting, thus making the subrecursive restriction sound  for general free algebras. 

The rest of this paper is structured as follow:
\begin{varitemize}
\item In Section~\ref{Sect:GenTieredRec} we define the class of the functions generated by general tiered recursion.
\item In Section~\ref{Sect:GraphRep} graph rewriting is introduced and explained. Basic definition and fundamental properties are 
  given.
\item Section~\ref{Sec:TieredRecGraph} is devoted to the main technical results of the paper: tiered recursion is realized by term 
  graph rewriting and reduction can be performed in polynomial time.
\item In Section~\ref{Sec:Conclusion} we state some conclusions and  some final considerations about our work.
\end{varitemize}

\section{General Tiered Recursion}
\label{Sect:GenTieredRec}
A \emph{signature} $\sigone$ is a pair $(\funsetone,\arityone)$ where
$\funsetone$ is a set of symbols and $\arityone:\funsetone\rightarrow\N$
assigns to every symbol $\funone\in\funsetone$ an arity $\arityone(\funone)$.
Given two signatures $\sigone=(\funsetone,\arityone)$
and $\sigtwo=(\funsettwo,\aritytwo)$, we write
$\sigtwo\subsig\sigone$ iff $\funsettwo\subseteq\funsetone$ and $\arityone(\funone)=\aritytwo(\funone)$
for every $\funone\in\funsettwo$.
The set of \emph{terms} over a signature $\sigone$ can be easily defined
by induction as the smallest set $\terms{\sigone}$ such that:
\begin{varitemize}
\item
  If $\funone\in\funsetone$ and $\arityone(\funone)=0$, then $\funone$ itself,
  seen as an expression, is in $\terms{\sigone}$;
\item
  If $\funone\in\funsetone$, $\arityone(\funone)=n\geq 1$ and $\termone_1,\ldots,\termone_n\in\terms{\sigone}$,
  then the expression $\funone(\termone_1,\ldots,\termone_n)$ is in $\terms{\sigone}$.
\end{varitemize}
Given a signature $\sigone=(\{\funone_1,\ldots,\funone_t\},\arityone)$ and natural numbers $n,i_1,\ldots,i_n,i$, the classes of
functions defined by \emph{tiered recursion (on $\sigone$) with tiers $(i_1,\ldots,i_n)\rightarrow i$} are the smallest
collections of functions from $\terms{\sigone}^n$ to $\terms{\sigone}$ satisfying the following conditions:
\begin{varitemize}
\item
  For every $i\in\N$, the identity function $\id:\terms{\sigone}\rightarrow\terms{\sigone}$ is tiered recursive with tiers
  $i\rightarrow i$.
\item
  For every $i\in\N$ and for every $\funone_j$,
  the constructor function $\constr{\funone_j}:\terms{\sigone}^{\arityone(\funone_j)}\rightarrow\terms{\sigone}$
  is tiered recursive with tiers $(i,\ldots,i)\rightarrow i$.
\item
  For every $i,i_1,\ldots,i_n\in\N$ and for every $1\leq m\leq n$, the projection function
  $\proj{n}{m}:\terms{\sigone}^{n}\rightarrow\terms{\sigone}$
  is tiered recursive with tiers $(i_1,\ldots,i_n)\rightarrow i$ whenever $i_m=i$.
\item
  For every $i,i_1,\ldots,i_n,j_1,\ldots,j_m\in\N$, if 
  $\tfunone:\terms{\sigone}^n\rightarrow\terms{\sigone}$ is tiered recursive
  with tiers $(i_1,\ldots,i_n)\rightarrow i$ and for every $1\leq k\leq n$ the
  function $\tfuntwo_k:\terms{\sigone}^m\rightarrow\terms{\sigone}$
  is tiered recursive with tiers $(j_1,\ldots,j_m)\rightarrow i_k$, then
  the composition $\comp{\tfunone,\tfuntwo_1,\ldots,\tfuntwo_n}:\terms{\sigone}^m\rightarrow\terms{\sigone}$ 
  (defined in the obvious way) is tiered recursive with tiers $(j_1,\ldots,j_m)\rightarrow i$.
\item
  For every $i,j,i_1,\ldots,i_n\in\N$ with $i<j$, if for every $\funone_k$ there is
  a function $\tfuntwo_k:\terms{\sigone}^{2\arityone(\funone_i)+n}\rightarrow\terms{\sigone}$ which is tiered recursive
  with tiers 
  $$
  (\overbrace{j,\ldots,j}^{\mbox{$\arityone(\funone_k)$ times}},\overbrace{i,\ldots,i}^{\mbox{$\arityone(\funone_k)$ times}},i_1,\ldots,i_n)\rightarrow i
  $$ 
  then the function $\tfunone=\rec{\tfuntwo_1,\ldots,\tfuntwo_t}:\terms{\sigone}^{n+1}\rightarrow\terms{\sigone}$ 
  defined by primitive recursion as
  \begin{align*}
    \tfunone(&\funone_i(\termone_1,\ldots,\termone_{\arityone(\funone_i)}),\termtwo_1,\ldots,\termtwo_n)=\\
      &\tfuntwo_i(\termone_1,\ldots,\termone_{\arityone(\funone_i)},\tfunone(\termone_1,\termtwo_1,\ldots,\termtwo_n),\ldots,
      \tfunone(\termone_{\arityone(\funone_i)},\termtwo_1,\ldots,\termtwo_n),\termtwo_1,\ldots,\termtwo_n)
  \end{align*}
  is tiered recursive with tiers $(j,i_1,\ldots,i_n)\rightarrow i$
\item
  For every $i,j,i_1,\ldots,i_n\in\N$, if for every $\funone_i$ there is
  a function $\tfuntwo_i:\terms{\sigone}^{\arityone(\funone_i)+n}\rightarrow\terms{\sigone}$ which is tiered recursive
  with tiers 
  $$
  (\overbrace{j,\ldots,j}^{\mbox{$\arityone(\funone_i)$ times}},i_1,\ldots,i_n)\rightarrow i
  $$ 
  then the function $\tfunone=\cond{\tfuntwo_1,\ldots,\tfuntwo_t}:\terms{\sigone}^{n+1}\rightarrow\terms{\sigone}$ 
  defined as
  $$
    \tfunone(\funone_i(\termone_1,\ldots,\termone_{\arityone(\funone_i)}),\termtwo_1,\ldots,\termtwo_n)=
      \tfuntwo_i(\termone_1,\ldots,\termone_{\arityone(\funone_i)},\termtwo_1,\ldots,\termtwo_n)
  $$
  is tiered recursive with tiers $(j,i_1,\ldots,i_n)\rightarrow i$.
\end{varitemize}
In the following, metavariables like $\tiersone$ and $\tierstwo$ will be used for expressions
like $(i_1,\ldots,i_n)\rightarrow i$.

Roughly speaking, the r\^ole of tiers is to single out ``a copy'' of the signature by a level: this 
level permits to control the nesting of the recursion. Note that function composition preserves tiers, 
i.e. we can substitute terms only for variables of equal tier. Moreover, recursion is allowed 
only on a variable of tier higher than the tier of the function (in the definition, $i<j$ is required).
This construction comes from a \emph{predicative} notion of recurrence.
Examples of terms and functions follow. 

\begin{example}
\begin{varitemize}
\item 
  Let $\sigone_1$ be the signature $(\{\funone_1, \funone_2\}, \arityone)$. If $\arityone(\funone_1)=0$ 
  and $\arityone(\funone_2)=1$, then $\terms{\sigone_1}$ is in bijective correspondence with $\N$.
\item 
  Let $\sigone_2$ be the signature $(\{\funone_1, \funone_2\}, \arityone)$. If we take 
  $\arityone(\funone_1)=0$ and 
  $\arityone(\funone_2)=2$ then $\terms{\sigone_2}$ can be thought as the set of (unlabeled) binary trees.
\item 
  Let $\sigone_3$ be the signature $(\{\funone_1,\funone_2,\funone_3,\funone_4\}, \arityone)$ where 
  $\arityone(\funone_1)=\arityone(\funone_3)=0$ and $\arityone(\funone_2)=\arityone(\funone_4)=2$.
  Then $\terms{\sigone_3}$ is the set of binary trees with binary labels.
\item 
  The function $\textsf{sum}:\terms{\sigone_1}^2\rightarrow\terms{\sigone_1}$
  computing addition on natural numbers can 
  be defined as 
  \begin{align*}
    \mathsf{sum}(\funone_1,x)&=x;\\
    \mathsf{sum}(\funone_2(y),x)&=\funone_2(\mathsf{sum}(y,x)).
  \end{align*}
  It can be easily proved to be tiered recursive with tiers $(j,i)\rightarrow i$ 
  whenever $j>i$.
\item 
  We can define the function $\mathsf{mirror}$ on $\terms{\sigone_2}$ 
  which mirrors a tree (i.e. each left subtree becomes right subtree and viceversa):  
  \begin{align*}
    \mathsf{mirror}(\funone_1)&=\funone_1;\\
    \mathsf{mirror}(\funone_2(x,y))&=\funone_2(\mathsf{mirror}(y),\mathsf{mirror}(x));
  \end{align*}
  $\mathsf{mirror}$ can be proved to be tiered recursive with tiers $(j)\rightarrow i$
  whenever $j>i$.
\end{varitemize}
\end{example}

\section{Graph Rewriting}
\label{Sect:GraphRep}
In this section, we introduce term graph rewriting, following~\cite{TGRbarendregt} but
adapting the framework to our specific needs.
\begin{definition}[Labelled Graph]
Given a signature $\sigone=(\funsetone,\arityone)$, a 
\emph{labelled graph over $\sigone$} consists of a directed
acyclic graph together with an ordering on the outgoing edges of each vertex and a (partial)
labelling of vertices with symbols from $\sigone$ such that the out-degree of each node
matches the arity of the corresponding symbols (and is $0$ if the labelling is undefined).
Formally, a labelled graph is a triple $\tgone=(\vsone,\ordone,\labelone)$
where: 
\begin{varitemize}
\item
  $\vsone$ is a set of \emph{vertices}.
\item
  $\ordone:\vsone\rightarrow\vsone^*$ is a (total) \emph{ordering function}.
\item
  $\labelone:\vsone\rightharpoonup\funsetone$ is a (partial) \emph{labelling function} such
  that the length of $\ordone(\verone)$ is the arity of $\labelone(\verone)$ if
  $\labelone(\verone)$ is defined and is $0$ otherwise.
\end{varitemize}
A labelled graph $(\vsone,\ordone,\labelone)$ is \emph{closed} iff $\labelone$ is a 
total function. 
\end{definition}
Consider the signature $\Sigma=(\{\funone,\funtwo,\funthree,\funfour\},\arityone)$, where
arities assigned to $\funone,\funtwo,\funthree,\funfour$ by $\arityone$ are $2$, $1$, $0$, $2$. Examples of labelled graphs over 
the signature $\Sigma$ are the following ones:
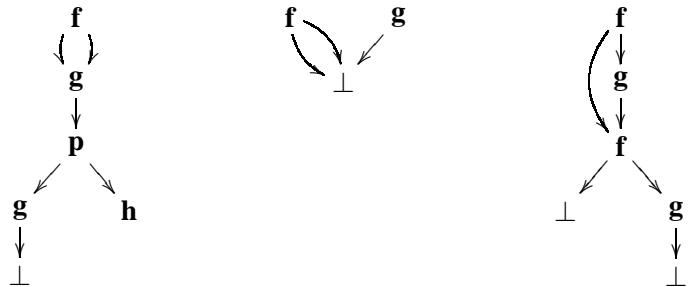
\begin{displaymath}
\xymatrix@R=11pt@C=8pt{
     & \funone \ar@/_/[d]\ar@/^/[d] &   \\
     & \funtwo \ar[d] &   \\
     & \funfour \ar[dl]\ar[dr] &   \\
\funtwo \ar[d] &          & \funthree \\
\bot &          &   \\
}
\hspace{50pt}
\xymatrix@R=11pt@C=8pt{
\funone\ar@/^/[dr]\ar@/_/[dr] & & \funtwo \ar[dl]\\
          & \bot &                         \\
}
\hspace{50pt}
\xymatrix@R=11pt@C=8pt{
     & \funone\ar@/_1pc/[dd]\ar[d] &   \\
     & \funtwo \ar[d] &   \\
     & \funone \ar[dl]\ar[dr] &   \\
\bot &          & \funtwo\ar[d] \\
     &          & \bot  \\
}
\end{displaymath}
The symbol $\bot$ denotes vertices where the underlying labelling function
is undefined (and, as a consequence, no edge departs from such vertices).
Their r\^ole is similar to the one of variables in terms.

If one of the vertices of a labelled graph is selected as the \emph{root}, we obtain
a term graph:
\begin{definition}[Term Graphs]
A \emph{term graph}, is a quadruple 
$\tgone=(\vsone,\ordone,\labelone,\rootone)$, where 
$(\vsone,\ordone,\labelone)$ is a labelled graph and
$\rootone\in\vsone$ is the \emph{root} of the term graph.
\end{definition}
The following are graphic representations of some term graphs. 
The root is the only vertex drawn inside a circle.
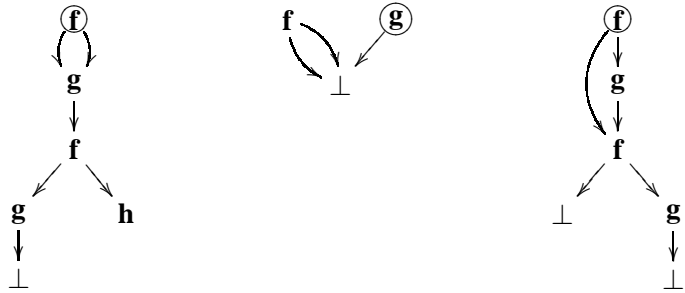
\begin{displaymath}
\xymatrix@R=11pt@C=8pt{
     &  *+[o][F]{\funone} \ar@/_/[d]\ar@/^/[d] &   \\
     & \funtwo \ar[d] &   \\
     & \funone \ar[dl]\ar[dr] &   \\
\funtwo \ar[d] &          & \funthree \\
\bot &          &   \\
}
\hspace{50pt}
\xymatrix@R=11pt@C=8pt{
 \funone\ar@/^/[dr]\ar@/_/[dr] & &  *+[o][F]{\funtwo} \ar[dl]\\
          & \bot &                         \\
}
\hspace{50pt}
\xymatrix@R=11pt@C=8pt{
     & *+[o][F]{\funone}\ar@/_1pc/[dd]\ar[d] &   \\
     & \funtwo \ar[d] &   \\
     & \funone \ar[dl]\ar[dr] &   \\
\bot &          & \funtwo\ar[d] \\
     &          & \bot  \\
}
\end{displaymath}
Given a (closed) term graph $\tgone$ on $\sigone$, $\GtoT{\tgone}$ is simply
the term in $\terms{\sigone}$ obtained by unfolding $\tgone$ starting
from its root.

The notion of an homomorphism between labelled graphs is not only
interesting mathematically, but will be crucial in defining rewriting:
\begin{definition}[Homomorphisms]
\begin{sloppypar}
An \emph{injective homomorphism} between two labelled graphs
$\tgone=(\vsone_\tgone,\ordone_\tgone,\labelone_\tgone)$ and 
$\tgtwo=(\vsone_\tgtwo,\ordone_\tgtwo,\labelone_\tgtwo)$ over the same 
signature $\sigone$ is a function $\homone$ from $\vsone_\tgone$ to $\vsone_\tgtwo$ 
preserving the labelled graph structure. In particular
\end{sloppypar}
\begin{align*}
  \labelone_\tgtwo(\homone(\verone))&=&\labelone_\tgone(\verone);\\
  \ordone_\tgtwo(\homone(\verone))&=&\homone^*(\ordone_\tgone(\verone));
\end{align*}
for any $\verone\in\domain{\labelone_\tgone}$, where $\homone^*$ is the obvious
generalization of $\homone$ to sequences of vertices.
Moreover, $\homone$ is injective on $\domain{\labelone_\tgone}$, i.e.,
$\homone(\verone)=\homone(\vertwo)$ implies $\verone=\vertwo$ for every
$\verone,\vertwo\in\domain{\labelone_\tgone}$.
An \emph{injective homomorphism} between two term graphs 
$\tgone=(\vsone_\tgone,\ordone_\tgone,\labelone_\tgone,\rootone_\tgone)$ and 
$\tgtwo=(\vsone_\tgtwo,\ordone_\tgtwo,\labelone_\tgtwo,\rootone_\tgtwo)$ is
an injective homomorphism between $(\vsone_\tgone,\ordone_\tgone,\labelone_\tgone)$ and 
$(\vsone_\tgtwo,\ordone_\tgtwo,\labelone_\tgtwo)$ such that
$\homone(\rootone_\tgone)=\rootone_\tgtwo$. Two labelled graphs $\tgone$ and $\tgtwo$ are 
isomorphic iff there is a bijective homomorphism from
$\tgone$ to $\tgtwo$; in this case, we write $\tgone\cong\tgtwo$. Similarly
for term graphs.
\end{definition}

In the sequel, we will always use \emph{homomorphism}, to mean \emph{injective homomorphism}.
Injectivity of $\homone$ on labelled vertices is not part of the
usual definition of an homomorphism between labelled graphs (see~\cite{TGRbarendregt}). We insist on injectivity because we want a rewriting rule  ``to match without sharing'' (see the notion of redex, Definition~\ref{Def:Redex}), which will be crucial for the r\^ole of unfolding graph rewriting rules (Section~\ref{Sect:Unfolding}) in the implementation of tiered recursion. Injectivity makes our notion of redex less general than in the usual setting
of graph rewriting. This, however, suffices for our purposes.
In the following, we will consider term graphs modulo isomorphism, i.e., $\tgone=\tgtwo$
iff $\tgone\cong\tgtwo$. Observe that two isomorphic term graphs have the same graphical
representation.
\begin{definition}[Graph rewriting Rules]
A \emph{graph rewriting rule} over a signature $\sigone$ 
is a triple $\rrone=(\tgone,\rootone,\roottwo)$ such that:
\begin{varitemize}
\item
  $\tgone$ is a labelled graph;
\item
  $\rootone,\roottwo$ are vertices of $\tgone$, called 
  the \emph{left root} and the \emph{right root} of $\rrone$,
  respectively.
\end{varitemize}
\end{definition}
The following are three examples of graph rewriting rules, assuming $\funone,\funtwo,\funthree,\funfour$ to be function
symbols in the underlying signature $\sigone$:
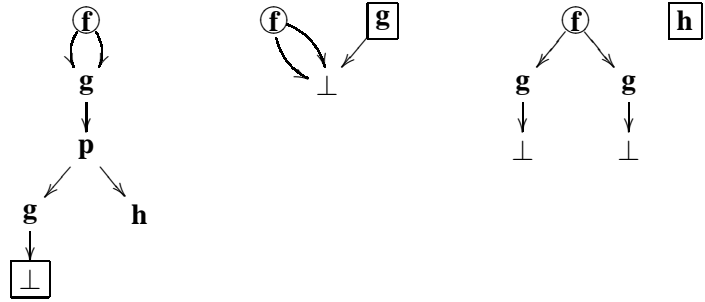
\begin{displaymath}
\xymatrix@R=11pt@C=8pt{
     &  *+[o][F]{\funone} \ar@/_/[d]\ar@/^/[d] &   \\
     & \funtwo \ar[d] &   \\
     & \funfour \ar[dl]\ar[dr] &   \\
\funtwo \ar[d] &          & \funthree \\
*+[F]{\bot} &          &   \\
}
\hspace{40pt}
\xymatrix@R=11pt@C=8pt{
 *+[o][F]{\funone}\ar@/^/[dr]\ar@/_/[dr] & &  *+[F]{\funtwo} \ar[dl]\\
          & \bot &                         \\
}
\hspace{40pt}
\xymatrix@R=11pt@C=8pt{
        & *+[o][F]{\funone}\ar[dl]\ar[dr] &  &  *+[F]{\funthree}\\
\funtwo\ar[d] &                  & \funtwo\ar[d] &   \\
\bot    &                  & \bot    &   \\
}
\end{displaymath}
Graphically, the left root is the (unique) node inside a circle, while the right root is
the (unique) node inside a square.
\begin{definition}[Subgraphs]
Given a labelled graph $\tgone=(\vsone_\tgone,\ordone_\tgone,\labelone_\tgone)$ 
and any vertex $\verone\in\vsone_\tgone$, the \emph{subgraph of $\tgone$ rooted
at $\verone$}, denoted $\subgr{\tgone}{\verone}$, is the
term graph $(\vsone_{\subgr{\tgone}{\verone}},\ordone_{\subgr{\tgone}{\verone}},\labelone_{\subgr{\tgone}{\verone}},
\rootone_{\subgr{\tgone}{\verone}})$ where
\begin{varitemize}
\item
  $\vsone_{\subgr{\tgone}{\verone}}$ is the subset of $\vsone_\tgone$ whose elements
  are vertices which are reachable from $\verone$ in $\tgone$.
\item
  $\ordone_{\subgr{\tgone}{\verone}}$ and $\labelone_{\subgr{\tgone}{\verone}}$ are
  the appropriate restrictions of $\ordone_\tgone$ and $\labelone_\tgone$
  to $\vsone_{\subgr{\tgone}{\verone}}$.
\item
  $\rootone_{\subgr{\tgone}{\verone}}$ is $\verone$.
\end{varitemize}
\end{definition}
A term graph $\tgone=(\vsone,\ordone,\labelone,\rootone)$ is said to be a 
\emph{proper term graph} if $\subgr{(\vsone,\ordone,\labelone)}{\rootone}=\tgone$.

We are finally able to give the notion of a redex, that represents
the occurrence of the lhs of a rewriting rule  in a graph:
\begin{definition}[Redexes]\label{Def:Redex}
Given a labelled graph $\tgone$, a \emph{redex} for $\tgone$ is
a pair $(\rrone,\homone)$, where $\rrone$ is a rewriting rule  
$(\tgtwo,\rootone,\roottwo)$ and $\homone$ is an homomorphism
between $\subgr{\tgtwo}{\rootone}$ and $\tgone$.
\end{definition}
If $((\tgtwo,\rootone,\roottwo),\homone)$ is a redex in
$\tgone$, we say, with a slight abuse of notation, that $\homone(\rootone)$ \emph{is} 
itself a redex. In most cases, this does not introduce any ambiguity.
Given a term graph $\tgone$ and a redex $((\tgtwo,\rootone,\roottwo),\homone)$,
the result of firing the redex is another term graph obtained by
successively applying the following three steps to $\tgone$:
\begin{varenumerate}
\item
  The \emph{build phase}: create an isomorphic copy of the portion of
  $\subgr{\tgtwo}{\roottwo}$ not contained in
  $\subgr{\tgtwo}{\rootone}$ (which may contain arcs originating in $\subgr{\tgtwo}{\roottwo}$ 
  and entering $\subgr{\tgtwo}{\rootone}$), and add it to $\tgone$, obtaining $\tgthree$.
  The underlying ordering and labelling functions are defined in the natural
  way.
\item
  The \emph{redirection phase}: all edges in $\tgthree$ pointing to $\homone(\rootone)$
  are replaced by edges pointing to the copy of $\roottwo$. If $\homone(\rootone)$
  is the root of $\tgone$, then the root of the newly created graph will be
  the newly created copy of $\roottwo$. The graph $\tgfour$ is obtained.
\item
  The \emph{garbage collection phase}: all vertices which are not accessible
  from the root of $\tgfour$ are removed. The graph $\tgfive$ is obtained.
\end{varenumerate}
We will write $\tgone\rewrite{(\tgtwo,\rootone,\roottwo)}\tgfive$ (or simply
$\tgone\rewrgraph\tgfive$, if this does not cause ambiguity) in this case.

\begin{example}
As an example, assuming again $\funone,\funtwo,\funthree$ to be
function symbols, consider the term graph $\tgone$ and the rewriting rule 
$\rrone=(\tgtwo,\rootone,\roottwo)$:
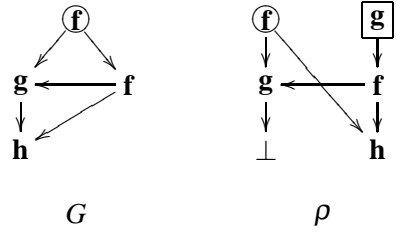
\begin{displaymath}
\xymatrix@R=11pt@C=8pt{
  &  *+[o][F]{\funone}\ar[dl]\ar[dr] &   \\
  \funtwo \ar[d] & & \funone \ar[ll]\ar[dll] \\
  \funthree & &  \\
  & \tgone & \\
}
\hspace{40pt}
\xymatrix@R=11pt@C=8pt{
 *+[o][F]{\funone} \ar[d]\ar[ddrr] & & *+[F]{\funtwo}\ar[d] \\
 \funtwo \ar[d] & & \funone\ar[ll]\ar[d] \\
\bot & & \funthree \\
 & \rrone & \\
}
\end{displaymath}
There is an homomorphism $\homone$ from
$\subgr{\tgtwo}{\rootone}$ to
$\tgone$. In particular, $\homone$ maps $\rootone$ to
the rightmost vertex in $\tgone$.
Applying the build phase and the redirection phase we get $\tgthree$
and $\tgfour$ as follows:
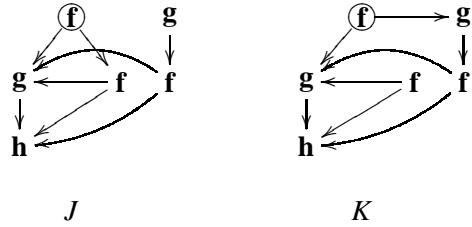
\begin{displaymath}
\xymatrix@R=11pt@C=8pt{
  &  *+[o][F]{\funone}\ar[dl]\ar[dr] & & \funtwo\ar[d]  \\
  \funtwo \ar[d] & & \funone \ar[ll]\ar[dll] & \funone\ar@/^/[dlll]\ar@/_1pc/[lll] \\
  \funthree & & & \\
  & \tgthree & & \\
}
\hspace{40pt}
\xymatrix@R=11pt@C=8pt{
  &  *+[o][F]{\funone}\ar[dl]\ar[rr] & & \funtwo\ar[d]  \\
  \funtwo \ar[d] & & \funone \ar[ll]\ar[dll] & \funone\ar@/^/[dlll]\ar@/_1pc/[lll] \\
  \funthree & & & \\
  & \tgfour & & \\
}
\end{displaymath}
Finally, applying the garbage collection phase, we get the
result of firing the redex $(\rrone,\homone)$:
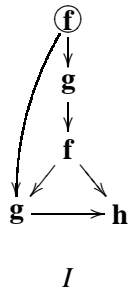
\begin{displaymath}
\xymatrix@R=11pt@C=8pt{
 & *+[o][F]{\funone}\ar@/_/[dddl]\ar[d] & \\
 & \funtwo\ar[d] & \\
 & \funone\ar[dl]\ar[dr] & \\
 \funtwo\ar[rr] & & \funthree\\
 & I & \\
}
\end{displaymath}
\end{example}
Given a proper graph $\tgone$ and a redex $(\rho, \phi)$ for $\tgone$, it 
is easy to prove that the result of firing the redex $(\rho, \phi)$ 
is a proper term graph: this is an immediate consequence of how the garbage
collection phase is defined.

\condinc{
We can define two restrictions on
$\rewrgraph$ as follows. Let $((\tgtwo,\rootone,\roottwo),\homone)$ be a redex
in $\tgone$. Then it is said to be
\begin{varitemize}
\item
  An \emph{innermost} redex iff for every redex 
  $((\tgthree,\rootthree,\rootfour),\homtwo)$ in $\tgone$,
  there is no proper path from $\homone(\rootone)$ to $\homtwo(\rootthree)$.
\item
  An \emph{outermost} redex iff for every redex 
  $((\tgthree,\rootthree,\rootfour),\homtwo)$ in $\tgone$,
  there is no proper path from $\homtwo(\rootthree)$ to $\homone(\rootone)$.
\end{varitemize}
If the redex $((\tgtwo,\rootone,\roottwo),\homone)$ is innermost
we also write $\tgone\rewritecbv{(\tgtwo,\rootone,\roottwo)}\tgfive$ or
$\tgone\rewrgraphcbv\tgfive$. Similarly, for an outermost redex
$((\tgtwo,\rootone,\roottwo),\homone)$ we write
$\tgone\rewritecbn{(\tgtwo,\rootone,\roottwo)}\tgfive$ or
$\tgone\rewrgraphcbn\tgfive$.
}
{
The notion of innermost and outermost graph rewriting can be defined in a natural way.
If $\tgone\rewrite{(\tgtwo,\rootone,\roottwo)}\tgfive$ by way of innermost graph
rewriting, we'll write $\tgone\rewritecbv{(\tgtwo,\rootone,\roottwo)}\tgfive$ (or
simply $\tgone\rewrgraphcbv\tgfive$). Similarly, for outermost reduction:
$\tgone\rewritecbn{(\tgtwo,\rootone,\roottwo)}\tgfive$ or
$\tgone\rewrgraphcbn\tgfive$.
}

Given two graph rewriting rules $\rrone=(\tgtwo,\rootone,\roottwo)$ and $\rrtwo=(\tgthree,\rootthree,\rootfour)$, 
$\rrone$ and $\rrtwo$ are said to be \emph{overlapping} iff there is a term graph $\tgone$ and two homomorphisms
$\homone$ and $\homtwo$ such that $(\rrone,\homone)$ and $(\rrtwo,\homtwo)$ are both redexes in $\tgone$
with $\homone(\rootone)=\homtwo(\rootthree)$. 
\begin{definition}
A \emph{graph rewriting system} (GRS) over a signature $\sigone$ is 
a set $\sgrone$ of non-overlapping graph rewriting rules on $\sigone$.
\end{definition}
If $\tgone\rewritecbn{\rrone}\tgtwo$ and $\rrone\in\sgrone$, we
write $\tgone\rewritecbn{\sgrone}\tgtwo$ or $\tgone\rewrgraphcbn\tgtwo$, if 
this doesn't cause any ambiguity. Similarly when $\tgone\rewritecbv{\rrone}\tgtwo$.

The notion of a term graph can be generalized into the notion of a \emph{multi-rooted}
term graph, i.e., a graph with $n\geq 1$ (not necessarily distinct) roots. Formally,
it is a tuple $\tgone=(\vsone,\ordone,\labelone,\rootone_1,\ldots,\rootone_n)$, where 
$(\vsone,\ordone,\labelone)$ is a labelled graph and
$\rootone_1,\ldots,\rootone_n\in\vsone$. Likewise, we can easily define the
subgraph of $\tgone$ \emph{rooted
at $\verone_1,\ldots,\verone_n$}, denoted $\subgr{\tgone}{\verone_1,\ldots,\verone_n}$,
as a multi-rooted term graph. Similarly for homomorphisms.
\condinc{
\subsection{Some Confluence Results}
We now want to give some confluence results for GRSs. Let us first focus
on outermost rewriting. Intuitively, outermost rewriting is the most efficient way of performing
reduction in presence of sharing, since computation is performed only if its result does
not risk to be erased. First, we need the following auxiliary lemma:
\begin{lemma}
Suppose $\tgone\rewrgraphcbn\tgtwo$ and $\tgone\rewrgraph\tgthree$, where $\tgtwo\neq\tgthree$.
Then either $\tgthree\rewrgraphcbn\tgtwo$
or there is $\tgfour$ such that $\tgtwo\rewrgraph\tgfour$ and $\tgthree\rewrgraphcbn\tgfour$.
\end{lemma}
\begin{proof}
Let $\verone$ and $\vertwo$ be the two redexes in $\tgone$ giving rise
to $\tgone\rewrgraphcbn\tgtwo$ and $\tgone\rewrgraph\tgthree$, respectively.
Similarly, let $\rrone$ and $\rrtwo$ be the two involved rewriting rules.
Clearly, there cannot be any (non-trivial) path from $\vertwo$ to $\verone$, by
definition of outermost rewriting. Now, the two rewriting steps are independent
from each other (because of the non-overlapping condition). 
There are now two cases. Either $w$ is erased when performing $\tgone\rewrgraphcbn\tgtwo$,
or it is not erased. In the first case, $w$ must be ``contained'' in $v$, and therefore,
we may apply $\rrone$ to $\tgthree$, obtaining $\tgtwo$. If $\vertwo$ has not been erased, one can clearly apply $\rrone$ to $\tgthree$ and
$\rrtwo$ to $\tgtwo$, obtaining a fourth graph $\tgfour$.\qed\end{proof}

The observation we just made can be easily turned into a more general result on reduction sequences of arbitrary length:
\begin{proposition}\label{prop:confcbn}
Suppose that $\tgone\rewrgraphcbn^n\tgtwo$ and $\tgone\rewrgraph^m\tgthree$. Then there are
$\tgfour$ and $k,l\in\N$ such that $\tgtwo\rewrgraph^k\tgthree$,
$\tgthree\rewrgraphcbn^l\tgfour$ and $n+k\leq m+l$. 
\end{proposition}
\begin{proof}
An easy induction on $n+m$.\qed
\end{proof}
Proposition~\ref{prop:confcbn} tells us that if we perform $n$ outermost steps and $m$ 
generic steps from $\tgone$, we can close the diagram in such a way that the number of
steps in the first branch is smaller or equal to the number of steps in the second branch.

With innermost reduction, the situation is exactly dual: 
\begin{lemma}
Suppose $\tgone\rewrgraph\tgtwo$ and $\tgone\rewrgraphcbv\tgthree$, where $\tgtwo\neq\tgthree$.
Then either $\tgthree\rewrgraph\tgtwo$
or there is $\tgfour$ such that $\tgtwo\rewrgraphcbv\tgfour$ and $\tgthree\rewrgraph\tgfour$.
\end{lemma}
\begin{proposition}\label{prop:confcbv}
Suppose that $\tgone\rewrgraph^n\tgtwo$ and $\tgone\rewrgraphcbv^m\tgthree$. Then there are
$\tgfour$ and $k,l\in\N$ such that $\tgtwo\rewrgraphcbv^k\tgthree$,
$\tgthree\rewrgraph^l\tgfour$ and $n+k\leq m+l$. 
\end{proposition}
In presence of sharing, therefore, outermost reduction is the best one can do, while
innermost reduction is the worst strategy, since we may reduce redexes in subgraphs 
that will be later discarded. As a by-product, we get confluence:
\begin{theorem}
Suppose that $\tgone\rewrgraphcbn^n\tgtwo$, $\tgone\rewrgraph^m\tgthree$ and
$\tgone\rewrgraphcbv^k\tgfour$, where $\tgtwo$, $\tgthree$ and $\tgfour$
are normal forms. Then $\tgtwo=\tgthree=\tgfour$ and  $n\leq m\leq k$.
\end{theorem}
\begin{proof}
From $\tgone\rewrgraphcbn^n\tgtwo$, $\tgone\rewrgraph^m\tgthree$ and
Proposition~\ref{prop:confcbn}, it follows that $n\leq m$ and that
$\tgtwo=\tgthree$. From $\tgone\rewrgraphcbn^m\tgthree$, $\tgone\rewrgraph^k\tgfour$ and
Proposition~\ref{prop:confcbv}, it follows that $m\leq k$.\qed
\end{proof}}{}
\subsection{Unfolding Graph Rewriting Rules}\label{Sect:Unfolding}
When computing a recursively defined function $\tfunone$ by graph rewriting,
we need to take advantage of sharing. In particular, if the recurrence
argument is a graph $\tgone$, the number of recursive calls generated
by calling $\tfunone$ on the graph $\tgone$ should be equal to the number of
vertices of $\tgone$, which can in turn be exponentially smaller than the size
of $\GtoT{\tgone}$. Unfortunately, this cannot be achieved by a finite set of graph rewriting
rules: distinct rules, called \emph{unfolding graph rewriting rules} are needed
for each possible argument to $\tfunone$. 

Let $\sigone=(\funsetone,\arityone)$ and
$\sigtwo=(\funsettwo,\aritytwo)$ be signatures such that
$S$ and $T$ are disjoint but in bijective correspondence. Let
$\varphi:\funsetone\rightarrow\funsettwo$ be a bijection and suppose
$\aritytwo(\varphi(\verone))=2\arityone(\verone)+n$ whenever $\verone$
is in $\funsetone$, for a fixed $n\in\N$. Finally, let $\funone$
be a symbol not in $\funsetone\cup\funsettwo$.

Under these hypothesis, an \emph{unfolding graph rewriting rule
for $\sigone$, $\sigtwo$ and $\funone$} is a graph rewriting
rule $\rrone=(\tgone,\rootone,\roottwo)$
where $\tgone=(\vsone,\ordone,\labelone)$ is a labelled graph
on a signature $\sigthree\supsig\sigone+\sigtwo$ assigning
arity $n+1$ to $\funone$, and satisfying the following constraints:
\begin{varitemize}
\item
  The elements of $\vsone$ are the (pairwise distinct) vertices
  $$
  \verone_1,\ldots,\verone_m,\vertwo_1,\ldots,\vertwo_m,\verthree_1,\ldots,\verthree_n,\verfour
  $$
  Let $\chi$ be a function mapping any $\verone_i$ to $\vertwo_i$.
\item
  For every $1\leq i\leq m$, $\ordone(\verone_i)$ is a sequence
  of vertices from $\{\verone_1,\ldots,\verone_m\}$. Moreover,
  the set of vertices of $\subgr{\tgone}{\verone_1}$ coincides
  with $\{\verone_1,\ldots,\verone_m\}$.
\item
  For every $1\leq i\leq m$, $\labelone(\verone_i)\in\funsetone$ and $\labelone(\vertwo_i)=\varphi(\labelone(\verone_i))$;
  moreover
  $$
  \ordone(\vertwo_i)=\ordone(\verone_i)\chi^*(\ordone(\verone_i))\verthree_1\ldots\verthree_n
  $$
\item
  $\labelone(\verfour)=\funone$ and
  $\ordone(\verfour)=\verone_1\verthree_1\ldots\verthree_n$.
\item
  $\labelone(\verthree_i)$ is undefined for every $i$.
\item
  $\rootone=\verfour$ and $\roottwo=\vertwo_1$.
\end{varitemize}

\begin{example}\label{ex:unfolding}
Let $\sigone=(\{\funtwo_1,\funtwo_2,\funtwo_3\},\arityone)$, where
arities assigned to $\funtwo_1,\funtwo_2,\funtwo_3$ are $2$, $1$, $0$,
respectively. Let $\sigtwo=(\{\funthree_1,\funthree_2,\funthree_3\},\arityone)$, where
arities assigned to $\funthree_1,\funthree_2,\funthree_3$ are $6$, $4$, $2$,
respectively. If $\sigthree\supsig\sigone+\sigtwo$ is a signature attributing
arity $3$ to $\funone$, we are in a position to give examples of unfolding
graph rewriting rules for $\sigone$, $\sigtwo$ and $\funone$. Here is one:
  
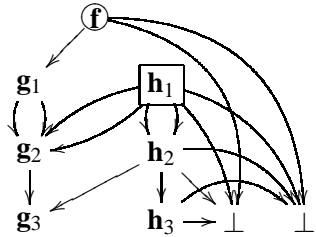
\begin{displaymath}
    \xymatrix@R=11pt@C=12pt{
      & *+[o][F]{\funone} \ar[ld]\ar@/^2pc/[rrddd]\ar@/^2.3pc/[rrrddd] & & & \\
      \funtwo_1\ar@/_/[d]\ar@/^/[d] & & *+[F]{\funthree_1}\ar@/^/[lld] \ar@/_/[lld] \ar@/_/[d]\ar@/^/[d]\ar@/^/[rdd]\ar@/^1pc/[rrdd] & & \\
      \funtwo_2\ar[d] & & \funthree_2\ar[lld]\ar[d]\ar[rd]\ar@/^1pc/[rrd] & & \\
      \funtwo_3 & & \funthree_3\ar[r]\ar@/^1.3pc/[rr] & \bot & \bot \\
    }    
  \end{displaymath}
\end{example}
{
Informally, thus, in an unfolding graph rewriting rule  for 
$\sigone$, $\sigtwo$ and $\funone$, we may single out four parts.
First, the root, labelled with $\funone$, which is also the left root of the rule;
second, a
subgraph labelled only on $\sigone$ (and this is $\subgr{\tgone}{\verone_1}$), which represents the recurrence
argument to $\funone$; third, a subgraph labelled only on $\sigtwo$ (it is $\subgr{\tgone}{\vertwo_1}$) which is isomorphic to the second
part, but for the addition of certain outgoing edges (the root of this part is the right root
of the rule); a last part consisting of $n$ unlabelled vertices, which have only incoming edges, coming from 
the root and $\subgr{\tgone}{\vertwo_1}$.
}

\subsection{Infinite Graph Rewriting Systems}
In~\cite{UgoSimoFopara}, two of the authors proved that \emph{finite} graph rewriting systems
are polynomially invariant when seen as a computational model. In other words, Turing
machines and finite GRSs can simulate each other with a polynomial overhead
in both directions, where both computational models are taken with their
natural cost models.

In next section, however, we will use \emph{infinite} GRSs (that is, GRSs with an infinite set of rules) to 
implement recurrence on arbitrary free algebras. We thus need some suitable
notion of computability on such infinite systems (which could be even uncomputable).
We  say that a specific graph rewriting system $\sgrone$ on the signature
$\sigone$ is \emph{polytime presentable} 
if there is a deterministic polytime algorithm $\algone_\sgrone$ which, given a term graph $\tgone$ on $\sigone$ returns:

\begin{varitemize}
\item
  A term graph $\tgtwo$ such that $\tgone\rewritecbv{\sgrone}\tgtwo$;
\item
  The value $\bot$ if such a graph $\tgtwo$ does not exist.
\end{varitemize}
In other words, a (infinite) GRS is polytime presentable iff there is a polytime algorithm 
which is able to compute any reduct of
any given graph, returning an error value if such a reduct does not exist.

\subsection{Graph Rewriting in Context}
A \emph{context} $\contone$ is simply a term graph $\tgone$. Given a context $\contone$,
let us denote with $\var{\contone}$ the set of those vertices of $\tgone$ which are not
labelled. 

Given a context $\contone$, another term graph $\tgone$ and a function $\subone$ mapping
every element of $\vsone\subseteq\var{\contone}$ into a vertex of $\tgone$, the
term graph $\subst{\contone}{\subone}{\tgone}$ is the one obtained from $\contone$
and $\tgone$ by removing all the vertices in $\vsone$ and
by redirecting to $\subone(\verone)$ every edge pointing to $\verone\in\vsone$.
More formally, given context $\contone=(\vsone,\ordone,\labelone,\rootone)$,
term graph $\tgone=(\vstwo,\ordtwo,\labeltwo,\roottwo)$ and function $\subone:\vsthree\rightarrow\vstwo$ 
such that $\vsthree\subseteq\var{\contone}$ and $\roottwo$ is in the range of 
$\subone$, $\subst{\contone}{\subone}{\tgone}$
is the term graph $(\vsfour,\ordthree,\labelthree,\rootthree)$ such that:
\begin{varenumerate}
\item
  The set of vertices $\vsfour$ is the disjoint union of $(\vsone-\vsthree)$ and $\vstwo$;
\item
  If $\verone\in\vsone-\vsthree$, then $\ordthree(\verone)$ is $\ordone(\verone)$
  where every occurrence of any $\vertwo\in\vsthree$ is replaced by $\subone(\vertwo)$;
\item
  If $\verone\in\vstwo$, then $\ordthree(\verone)=\ordtwo(\verone)$;
\item
  For every $\verone\in\vsone-\vsthree$, it holds that 
  $\labelthree(\verone)=\labelone(\verone)$, while for every $\verone\in\vstwo$,
  it holds that $\labelthree(\verone)=\labeltwo(\verone)$;
\item
  If $\rootone\in\vsthree$, then $\rootthree=\roottwo$, otherwise $\rootthree=\rootone$.
\end{varenumerate}

\begin{example}
Let $\contone$ and $\tgone$ be the following graphs
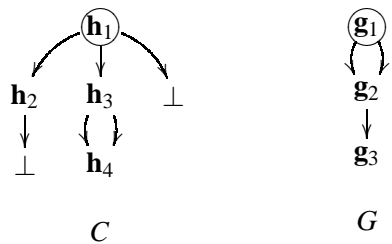
\begin{displaymath}
    \xymatrix@R=11pt@C=12pt{
  & & *+[o][F]{\funthree_1}\ar@/_/[ld]\ar[d]\ar@/^/[rd]& & \\
      & \funthree_2 \ar[d] & \funthree_3\ar@/_/[d]\ar@/^/[d] & \bot &  \\
      & \bot &  \funthree_4   & & \\
      &  & \contone & & \\
    }  
\hspace{40pt}
   \xymatrix@R=11pt@C=12pt{
        *+[o][F]{\funtwo_1}\ar@/_/[d]\ar@/^/[d] & &  & & \\
      \funtwo_2\ar[d] & & & & \\
      \funtwo_3 & & &  &  \\
      \tgone &  &  & & \\
    }  
\end{displaymath}
If we take the function $\xi$  such that $\xi$ maps the unlabelled node pointed by $\funthree_2$ to 
$\funtwo_1$ and the unlabelled node pointed by $\funthree_1$ to $\funtwo_3$, 
then the graph $\subst{\contone}{\subone}{\tgone}$ is 
  
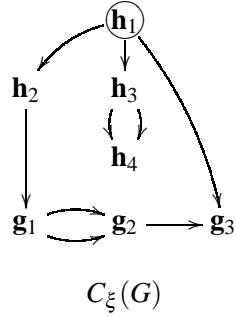
\begin{displaymath}
      \xymatrix@R=11pt@C=12pt{
  & & *+[o][F]{\funthree_1}\ar@/_/[ld]\ar[d]\ar@/^/[rddd]& & \\
      & \funthree_2 \ar[dd] & \funthree_3\ar@/_/[d]\ar@/^/[d] &  &  \\
  &  &  \funthree_4  & & \\
      &    \funtwo_1\ar@/^/[r]\ar@/_/[r]& \funtwo_2\ar[r] & \funtwo_3&&\\
      &  &\subst{\contone}{\subone}{\tgone}& & \\
    }  
  \end{displaymath}
\end{example}
When we write $\subst{\contone}{\subone}{\tgone}\rewrgraphcbv\subst{\contone}{\subtwo}{\tgtwo}$,
we are tacitly assuming that rewriting have taken place inside $\tgone$.
Notice that $\tgone\rewrgraphcbv\tgtwo$ does not imply
that $\subst{\contone}{\subone}{\tgone}\rewrgraphcbv\subst{\contone}{\subtwo}{\tgtwo}$ for some
$\subtwo$. Moreover, by the very definition of graph rewriting:
\begin{lemma}\label{lemma:iso}
If $\contone$ and $\tgone$ are proper and
$\subst{\contone}{\subone}{\tgone}\rewrgraphcbv\subst{\contone}{\subtwo}{\tgtwo}$,
then for every $\verone_1,\ldots,\verone_s$ such that 
$\subgr{\tgone}{\subone(\verone_1),\ldots,\subone(\verone_s)}$ does not contain
any redex, it holds that
$\subgr{\tgone}{\subone(\verone_1),\ldots,\subone(\verone_s)}
\cong\subgr{\tgtwo}{\subtwo(\verone_1),\ldots,\subtwo(\verone_s)}$.
\end{lemma}
In other words, those portions of $\tgone$ which do not contain any redex are preserved while
performing reduction in $\tgone$.

Contexts will be useful when proving that certain GRSs correctly computes tiered recursive
functions. In particular, they will allow us to prove those statements by induction on the
proof that the functions under consideration are tiered recursive.

\section{Implementing Tiered Recursion by Term Graph Rewriting}\label{Sec:TieredRecGraph}
Given a signature $\sigone=(\funsetone,\arityone)$, $\tiered{\sigone}$ stands
for the (infinite) signature
$$
(\{\funone^i\mid\funone\in\funsetone\mbox{ and }i\in\N\},\aritytwo)
$$
where $\aritytwo(\funone^i)=\arityone(\funone)$ for every $i\in\N$.
Given a term $\termone$ in $\terms{\sigone}$ and $i\in\N$, $\termone^i$ denotes
the term in $\terms{\tiered{\sigone}}$ obtained by labelling any
function symbol in $\termone$ with the specific natural number $i$; $\ptiered{\sigone}{i}$ is
the subsignature of $\tiered{\sigone}$ of those function symbols labelled with the
particular natural number $i$. With $\psize{\tgone}{i}$ we denote the number
of vertices of $\tgone$ labelled with functions in $\ptiered{\sigone}{i}$, whenever
$\tgone$ is a term graph on $\sigtwo\supsig\tiered{\sigone}$.

Suppose $\tfunone:\terms{\sigone}^n\rightarrow\terms{\sigone}$ and
suppose the term graph rewriting system $\sgrone$ over a signature
$\sigtwo\supsig\tiered{\sigone}$ (including a symbol $\tfunone$ of arity $n$)
is such that whenever $\tfunone(\termone_1,\ldots,\termone_n)=\termtwo$,
and $\GtoT{\tgone}=\funone(\termone_1^{i_1},\ldots,\termone_n^{i_n})$,
it holds that $\tgone\rewrgraphcbv^*\tgtwo$ where 
$\GtoT{\tgtwo}=\termtwo^i$. Then we say \emph{$\sgrone$ represents
$\tfunone$ with respect to $(i_1,\ldots i_n)\rightarrow i$}.

Now, let $\sgrone$ be a term graph rewriting system representing
$\tfunone$  with respect to $(i_1,\ldots i_n)\rightarrow i$ 
and let $p:\N\rightarrow\N$ be a polynomial. We say
that \emph{$\sgrone$ is bounded by $p$} iff whenever 
$\GtoT{\tgone}=\funone(\termone_1^{i_1},\ldots,\termone_n^{i_n})$
and $\tgone\rewrgraphcbv^m\tgtwo$, it holds that
$m,\size{\tgtwo}\leq p(\size{\tgone})$.

The main result of this paper is the following:
\begin{theorem}\label{theo:mainresult}
For every signature $\sigone$ and for every tiered recursive function 
$\tfunone:\terms{\sigone}^n\rightarrow\terms{\sigone}$
with tiers $\tiersone=(i_1,\ldots i_n)\rightarrow i$ there are a 
term graph rewriting system $\sgrone_\tfunone^\tiersone$ on 
$\sigtwo\supsig\tiered{\sigone}$ and a polynomial
$p:\N^n\rightarrow\N$ such that $\sgrone_\tfunone^\tiersone$ represents $\tfunone$ with respect to $\tiersone$,
being bounded by $p$. Moreover, $\sgrone_\tfunone^\tiersone$ is polytime presentable.
\end{theorem}

In other words, every tiered recursive function is represented by a GRS which is
potentially infinite, but which is polytime presentable. Moreover, appropriate
polynomial bounds hold for the number of innermost rewriting steps necessary to
compute the normal form of term graphs and for the size of any intermediate results
produced during computation.

In the rest of this section, we will give a proof of Theorem~\ref{theo:mainresult}.
This will be a constructive proof, i.e. we define $\sgrone_\tfunone^\tiersone$ by induction
on the structure of $\tfunone$ as a tiered function (i.e. on the structure of the
proof that $\tfunone$ is a tiered function). Let $\sigone=(\{\conone_1,\ldots,\conone_t\},\arityone)$.
$\sgrone_\tfunone^\tiersone$ is defined as follows:
\begin{varitemize}
\item
  For every $\tiersone=(i)\rightarrow i$, $\sgrone_{\id}^\tiersone$ is a GRS whose only rule is
  \begin{displaymath}
    \xymatrix@R=15pt@C=8pt{
      *+[o][F]{\funid^\tiersone}\ar[d]\\
      *+[F]{\bot}\\
    }
  \end{displaymath} 
\item
  For every $\tiersone=(i,\ldots,i)\rightarrow i$, $\sgrone_{\constr{\conone_j}}^\tiersone$ is a GRS whose only rule is
  \begin{displaymath}
    \xymatrix@R=15pt@C=8pt{
      *+[o][F]{\funconstr{\conone_j}{\tiersone}}\ar[d]\ar[rrd] & &  *+[F]{\conone_j^{i}} \ar[lld]\ar[d]\\
      \bot & \cdots & \bot\\
    }
  \end{displaymath}
\item
  For every $\tiersone$, $\sgrone_{\proj{n}{m}}^\tiersone$ is a GRS whose only rule is
  \begin{displaymath}
    \xymatrix@R=.3pt@C=8pt{
       & & & *+[o][F]{\tproj{n}{m}^\tiersone}\ar[llldddddd]\ar[ldddddd]\ar[dddddd]\ar[rdddddd]\ar[rrrdddddd] & & & \\
       & & & & & & \\
       & & & & & & \\
       & & & & & & \\
       & & & & & & \\
       & & & & & & \\
      \bot & \cdots & \bot & *+[F]{\bot} & \bot & \cdots & \bot\\
      \ar@{-} `d^r[r]`_d[r] & & \ar@{-} `d_l[l]`^d[l] & & \ar@{-} `d^r[r]`_d[r] & & \ar@{-} `d_l[l]`^d[l] \\
       & & & & & & \\
       & & & & & & \\
       & & & & & & \\
       & \mbox{$m-1$ times} &  &  &  & \mbox{$n-m$ times} & \\
   }
  \end{displaymath}
\item
  Let $\tfunthree$ be $\comp{\tfunone,\tfuntwo_1,\ldots,\tfuntwo_n}$ and suppose
  $\tfunthree$ is tiered recursive with tiers $\tiersone$. Then 
  $\sgrone_{\tfunthree}^\tiersone$ is the GRS $\sgrone_{\tfunone}^\tierstwo\cup\sgrone_{\tfuntwo_1}^{\tiersthree_1}
  \cup\ldots\cup\sgrone_{\tfuntwo_n}^{\tiersthree_n}\cup\{\rrone\}$, where
  $\rrone$ is the rule
  
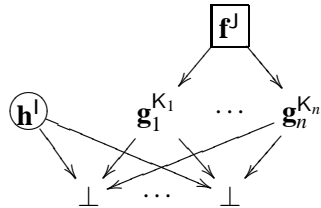
\begin{displaymath}
    \xymatrix@R=15pt@C=8pt{
      & & & *+[F]{\funone^\tierstwo}\ar[ld]\ar[rd] & \\
      *+[o][F]{\funthree^\tiersone} \ar[rd]\ar[rrrd] & & \funtwo_1^{\tiersthree_1}\ar[ld]\ar[rd] & \cdots & \funtwo_n^{\tiersthree_n}\ar[llld]\ar[ld] \\
      & \bot & \cdots & \bot & \\
    }    
  \end{displaymath}
  and $\tierstwo,\tiersthree_1,\ldots,\tiersthree_n$ are the tiers of
  $\tfunone,\tfuntwo_1,\ldots,\tfuntwo_n$, respectively.
\item
  Let $\tfunthree:\terms{\sigone}^{n+1}\rightarrow\terms{\sigone}$ be $\rec{\tfuntwo_1,\ldots,\tfuntwo_t}$ and
  suppose $\tfunthree$ is tiered recursive with tiers $\tiersone=(j,i_1,\ldots,i_n)\rightarrow i$. Let
  $\tierstwo_1,\ldots,\tierstwo_t$ be the tiers of $\tfuntwo_1,\ldots,\tfuntwo_t$, respectively. Moreover, let
  $\sigtwo=(\{\funtwo_1^{\tierstwo_1},\ldots,\funtwo_t^{\tierstwo_t}\},\aritytwo)$, where 
  $\aritytwo(\funtwo_i^{\tierstwo_i})=2\arityone(\funone_i)+n$.
  Then, $\sgrone_{\tfunthree}^\tiersone$ is the GRS $\sgrone_{\tfuntwo_1}^{\tierstwo_1}\cup\ldots\cup
  \sgrone_{\tfuntwo_t}^{\tierstwo_t}\cup\sgrtwo$
  where $\sgrtwo$ is the set of \emph{all} unfolding graph rewriting rules for $\ptiered{\sigone}{j}$, 
  $\sigtwo$ and $\funthree^\tiersone$.
\item
  Let $\tfunthree$ be $\cond{\tfuntwo_1,\ldots,\tfuntwo_t}$ and suppose $\tfunthree$ is tiered recursive with tiers 
  $\tiersone=(j,i_1,\ldots,i_n)\rightarrow i$. Then 
  $\sgrone_{\tfunthree}^\tiersone$ is the GRS $\sgrone_{\tfuntwo_1}^{\tierstwo_1}\cup\ldots\cup
  \sgrone_{\tfuntwo_t}^{\tierstwo_t}\cup\{\rrone_1,\ldots,\rrone_t\}$, where
  $\rrone_k$ is the rule
  
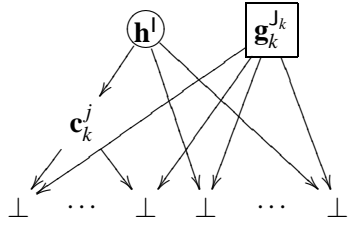
\begin{displaymath}
    \xymatrix@R=15pt@C=8pt{
      & & *+[o][F]{\funthree^\tiersone} \ar[ld]\ar[rdd]\ar[rrrdd] & & *+[F]{\funtwo_k^{\tierstwo_k}}\ar[lllldd]\ar[lldd]\ar[ldd]\ar[rdd]  & \\
      & \conone_k^j \ar[ld]\ar[rd] & & & & \\
      \bot & \cdots & \bot & \bot & \cdots & \bot \\
     }    
  \end{displaymath}
  and $\tierstwo_1,\ldots,\tierstwo_t$ are the tiers of $\tfuntwo_1,\ldots,\tfuntwo_t$, respectively.
\end{varitemize}
The \emph{extensional} soundness of the above encoding can be verified relatively easily. More interesting, and difficult,
is the study of its complexity properties.

Theorem~\ref{theo:mainresult} is a direct consequence of the following:
\begin{proposition}\label{prop:inductionload}
Suppose $\tfunone:\terms{\sigone}^{n}\rightarrow\terms{\sigone}$ is tiered recursive with tiers
$\tiersone=(i_1,\ldots,i_n)\rightarrow i$ and let $\sgrone_{\tfunone}$ be the GRS on $\sigtwo\supsig\tiered{\sigone}$
defined as above. Then there is a polynomial with natural coefficients $p:\N\rightarrow\N$ such that
for every proper context $\contone$, for every $\subone$ and for every
proper term graph $\tgone$ such that $\GtoT{\tgone}=\funone^\tiersone(\termone_1^{i_1},\ldots,\termone_n^{i_n})$, it holds
that
$\subst{\contone}{\subone}{\tgone}\rewrgraphcbv\subst{\contone}{\subtwo_1}{\tgtwo_1}
\rewrgraphcbv\ldots\rewrgraphcbv\subst{\contone}{\subtwo_m}{\tgtwo_m}$
where: 
\begin{varenumerate}
\item
  $\GtoT{\tgtwo_m}=(\tfunone(\termone_1,\ldots,\termone_n))^i$;
\item
  $m\leq p(\size{\tgone})$;
\item
  $\size{\tgtwo_j}\leq p(\size{\tgone})$ for every $j$;
\item
  $\psize{\tgtwo_m}{i}\leq\psize{\tgone}{i}+p(\sum_{k=i+1}^{\infty}\psize{\tgone}{k})$.
\end{varenumerate}
\end{proposition}
\begin{proof}
By induction on the structure of $\tfunone$ as a tiered function. In this proof, we use notations
like $\tgthree\rightarrow\tgfour$, meaning there is \emph{at least} an arc from $\tgthree$ to $\tgfour$,
$\tgthree\Rightarrow\tgfour$, meaning there is \emph{some} (but possibly zero) arcs 
from $\tgthree$ to $\tgfour$ and $\tgthree\dashrightarrow\tgfour$, meaning all the vertices 
in $\tgfour$ are reachable from $\tgthree$. Tiering information is omitted whenever possible, 
e.g., $\funone$ often takes the place of $\funone^\tiersone$. 
\condinc{}{We only give the most interesting inductive cases.}
\begin{varitemize}
\condinc{\item
  If $\tfunone=\id$, then $n=1$, $i_1=i$ and 
  $\subst{\contone}{\subone}{\tgone}\rewrgraphcbv\subst{\contone}{\subtwo_1}{\tgtwo_1}$,
  where $\GtoT{\tgtwo_1}=\termone_1^i=(\id(\termone_1))^i$. Moreover, $\psize{\tgtwo_1}{k}\leq\psize{\tgone}{k}$
  for every $k$. As a consequence, $m$ can be taken to be $1$ and $p$ can be taken to be
  $x\mapsto x+1$.
\item
  If $\tfunone=\constr{\conone_j}{i}$, then $i_k=i$ for every $k$ and
  $\subst{\contone}{\subone}{\tgone}\rewrgraphcbv\subst{\contone}{\subtwo_1}{\tgtwo_1}$,
  where $\GtoT{\tgtwo_1}=\conone^i_j(\termone_1^i,\ldots,\termone_n^i)=(\constr{\conone_j}{i}(\termone_1,\ldots,\termone_n))^i$.
  Moreover, $\psize{\tgtwo_1}{k}\leq\psize{\tgone}{k}$ for every $k\neq i$, while
  $\psize{\tgtwo_1}{i}\leq\psize{\tgone}{i}+1$.
  As a consequence, $m$ can be taken to be $1$ and $p$ can be taken to be
  the $x\mapsto x+1$.
\item
  If $\funone=\proj{n}{m}$ then $i_m=i$, and we have 
  $\subst{\contone}{\subone}{\tgone}\rewrgraphcbv\subst{\contone}{\subtwo_1}{\tgtwo_1}$ 
  where $\GtoT{\tgtwo_1}=(t_{m})^{i_m}=(\proj{n}{m}(\termone_1,\ldots,\termone_n))^{i_{m}}=(\proj{n}{m}(\termone_1,\ldots,\termone_n))^{i}$.
  We have that  $\psize{\tgtwo_1}{k}\leq\psize{\tgone}{k}$ for every $k$, then we take $p$ and $q$ respectively as $x\mapsto x$ and $x\mapsto 0$.}{}
\item
  Suppose $\tfunone=\comp{\tfuntwo,\tfunthree_1,\ldots,\tfunthree_k}$
  where $\tfuntwo$ is tiered recursive with tiers
  $(j_1,\ldots,j_k)\rightarrow i$ and $\tfunthree_h$ is tiered
  recursive with tiers $(i_1,\ldots,i_n)\rightarrow j_h$. Let
  $p_\tfuntwo,p_{\tfunthree_1},\ldots,p_{\tfunthree_k}$ be some polynomials
  satisfying the properties above, whose existence follows from
  the inductive hypothesis.
  We can write $\tgone$ as $\funone\rightarrow\tgthree$ and so 
  we start from
  
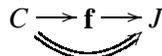
\begin{displaymath}
    \xymatrix@R=15pt@C=15pt{
      \contone \ar[r]\ar@/_1.2pc/@{=>}[rr] & \funone\ar[r] & \tgthree \\
    }
  \end{displaymath}
  In one rewriting step the graph becomes
  
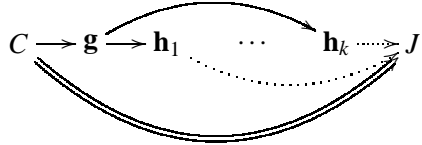
\begin{displaymath}
    \xymatrix@R=10pt@C=15pt{
      \contone \ar[r]\ar@/_3pc/@{=>}[rrrrr] & \funtwo\ar[r]\ar@/^1.3pc/[rrr] & 
      \funthree_1\ar@/_1.3pc/@{.>}[rrr] & \cdots & \funthree_k\ar@{.>}[r] & \tgthree \\
      & & & & & \\
      & & & & & \\
    }
  \end{displaymath}
  After some $m_1$ rewriting steps, we get to 
  
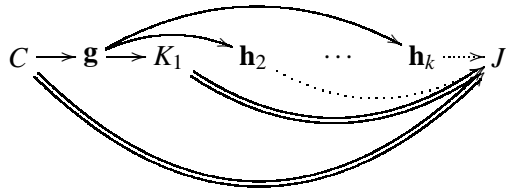
\begin{displaymath}
    \xymatrix@R=9pt@C=15pt{
      & & & & & & \\
      \contone \ar[r]\ar@/_4pc/@{=>}[rrrrrr] & \funtwo\ar[r]\ar@/^.7pc/[rr]\ar@/^1.6pc/[rrrr] & 
      \tgfour_1\ar@/_2pc/@{=>}[rrrr] & \funthree_2 \ar@/_1.3pc/@{.>}[rrr] & \cdots & \funthree_k\ar@{.>}[r] & \tgthree \\
      & & & & & & \\
      & & & & & & \\
      & & & & & & \\
    }
  \end{displaymath}
  By the induction hypothesis, the pointers coming from $\contone$ remain unaltered.
  Notice that $m_1\leq p_{\tfunthree_1}(\size{\tgthree}+1)\leq p_{\tfunthree_1}(\size{\tgone})$. Moreover, 
  $\tgfour_1$ can only contain vertices labelled with $\conone_k^{j_1}$, because the whole graph is still 
  proper. Then, by the inductive hypothesis,
  $$
  \psize{\tgfour_1}{i}\leq p_{\tfunthree_1}(\sum_{s=j_1+1}^{\infty}\psize{\tgone}{s})
  $$
  The size of any intermediate graph produced in these $m_1$ steps (not considering $\size{\contone}$)
  is $k+p_{\tfunthree_1}(\size{\tgthree}+1)\leq k+p_{\tfunthree_1}(\size{\tgone})$.
  Likewise, after $m_2+\ldots+m_{k-1}$ rewriting steps, again by induction hypothesis, we get to
  
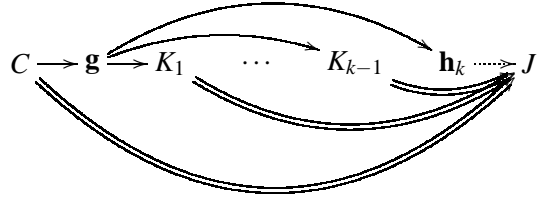
\begin{displaymath}
    \xymatrix@R=8pt@C=15pt{
      & & & & & & \\
      \contone \ar[r]\ar@/_4pc/@{=>}[rrrrrr] & \funtwo\ar[r]\ar@/^.9pc/[rrr]\ar@/^1.9pc/[rrrr] & 
      \tgfour_1\ar@/_2pc/@{=>}[rrrr] & \cdots & \tgfour_{k-1} \ar@/_.9pc/@{=>}[rr] & \funthree_k\ar@{.>}[r] & \tgthree \\
      & & & & & & \\
      & & & & & & \\
      & & & & & & \\
    }
  \end{displaymath}
  where for every $s$, 
  $m_s\leq p_{\tfunthree_s}(\size{\tgthree}+1)\leq p_{\tfunthree_s}(\size{\tgone})$, 
  $\psize{\tgfour_s}{t}=0$ whenever $t\neq j_s$ and 
  $$
  \psize{\tgfour_s}{j_s}\leq p_{\tfunthree_s}(\sum_{s=j_s+1}^{\infty}\psize{\tgone}{s})
  $$
  Moreover, the size of any intermediate graph produced in the $m_s$ steps (not considering $\size{\contone}$)
  is at most
  $$
  k+\sum_{r=1}^s p_{\tfunthree_r}(\size{\tgone}).
  $$
  After $m_k$ steps, we reach
  
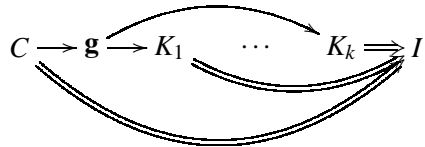
\begin{displaymath}
    \xymatrix@R=10pt@C=15pt{
      & & & & & \\
      \contone \ar[r]\ar@/_3pc/@{=>}[rrrrr] & \funtwo\ar[r]\ar@/^1.3pc/[rrr] & 
      \tgfour_1\ar@/_1.3pc/@{=>}[rrr] & \cdots & \tgfour_k\ar@{=>}[r] & \tgfive \\
      & & & & & \\
      & & & & & \\
    }
  \end{displaymath}
  Again, $m_k\leq p_{\tfunthree_k}(\size{\tgthree}+1)\leq p_{\tfunthree_k}(\size{\tgone})$,
  $\psize{\tgfour_k}{t}= 0$ whenever $t\neq j_k$ and
  $$
  \psize{\tgfour_k}{j_k}\leq p_{\tfunthree_k}(\sum_{s=j_k+1}^{\infty}\psize{\tgone}{s})
  $$
  This time, however, we cannot
  claim that $\tgthree$ remains unchanged. Indeed, it's replaced by $\tgfive$, which
  anyway only contains a subset of the vertices of $\tgthree$. 
  As usual, the size of any intermediate result is at most
  $$
  k+\sum_{r=1}^k p_{\tfunthree_k}(\size{\tgone}).
  $$
  The graph above can be written as follows:
  
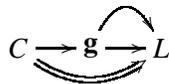
\begin{displaymath}
    \xymatrix@R=5pt@C=15pt{
      & & & & & \\
      \contone \ar[r]\ar@/_1pc/@{=>}[rr] & \funtwo\ar[r]\ar@/^1.3pc/[r] & \tgsix \\
      & & & & & \\
    }
  \end{displaymath}
  where $\tgsix$ is a graph such that
  $$
  \psize{\tgsix}{t}\leq r(\sum_{s=t+1}^\infty\psize{\tgone}{s})+\psize{\tgone}{t}
  $$
  for every $t$ and $r$ is a fixed polynomial not depending on $\tgone$.
  Finally, after $l$ steps, we get to $\contone\Rightarrow\tgseven$,
  where $\GtoT{\tgseven}=\tfuntwo(\termtwo_1,\ldots,\termtwo_k)^i$
  and the pointers coming from $\contone$ remain unaltered. Moreover:
  \begin{align*}
    l&\leq p_{\tfuntwo}(\size{\tgsix})\leq p_{\tfuntwo}(1+\sum_{t=1}^\infty\psize{\tgsix}{t})
      \leq p_{\tfuntwo}(1+\sum_{t=1}^\infty(r(\sum_{s=t+1}^\infty\psize{\tgone}{s})+\psize{\tgone}{t}))\\
      &\leq p_{\tfuntwo}(r(1+\sum_{t=1}^\infty\sum_{s=t}^\infty\psize{\tgone}{s}))
      \leq p_{\tfuntwo}(r(\max\{i_1,\ldots,i_n\}\size{\tgone}))=q(\size{\tgone})
  \end{align*}
  where $q$ is a polynomial. By induction hypothesis,
  \begin{align*}
    \psize{\tgseven}{i}&\leq p_{\tfuntwo}(\sum_{s=i+1}^\infty\psize{\tgsix}{s})+\psize{\tgsix}{i}
       \leq p_{\tfuntwo}(\sum_{s=i+1}^\infty  (r(\sum_{t=s+1}^\infty\psize{\tgone}{t})+\psize{\tgone}{s}))+\psize{\tgsix}{i}
       \leq p_{\tfuntwo}(\sum_{s=i+1}^\infty  (r(\sum_{t=s}^\infty\psize{\tgone}{t})))+\psize{\tgsix}{i}\\
       &\leq p_{\tfuntwo}(r(\sum_{s=i+1}^\infty\sum_{t=s}^\infty\psize{\tgone}{t}))+\psize{\tgsix}{i}
        \leq p_{\tfuntwo}(r(\max\{i_1,\ldots,i_n\}(\sum_{s=i+1}^\infty\psize{\tgone}{s})))+\psize{\tgsix}{i}\\
       &\leq p_{\tfuntwo}(r(\max\{i_1,\ldots,i_n\}(\sum_{s=i+1}^\infty\psize{\tgone}{s})))+r(\sum_{s=i+1}^\infty\psize{\tgone}{s})+\psize{\tgone}{i}
        =z(\sum_{s=i+1}^\infty\psize{\tgone}{s})+\psize{\tgone}{i}
  \end{align*}
  where $z$ is a polynomial. The size of intermediate results is itself bound by $q(\size{\tgone})$. We can choose
  $p_\tfunthree$ to be just $q+z$.
\item
  Suppose $\tfunone=\rec{\tfuntwo_1,\ldots,\tfuntwo_t}$
  where $\tfuntwo_i$ is tiered recursive with tiers
  $(i,\ldots,i,j,\ldots,j,i_1,\ldots,i_n)\rightarrow i$. Let
  $p_{\tfuntwo_1},\ldots,p_{\tfuntwo_n}$ be some polynomials
  satisfying the properties above, whose existence follows from
  the inductive hypothesis.
  We can write $\tgone$ as $\funone\Rightarrow\tgthree$ and so 
  we start from
  \begin{displaymath}
    \xymatrix@R=15pt@C=15pt{
      \contone \ar[r]\ar@/_1.2pc/@{=>}[rr] & \funone\ar@{=>}[r] & \tgthree \\
    }
  \end{displaymath}
  In one (unfolding) rewriting step the graph becomes
  
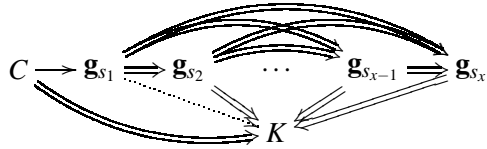
\begin{displaymath}
    \xymatrix@R=10pt@C=15pt{
      & & & & & \\
      \contone \ar[r]\ar@{=>}@/_1pc/[rrrd] & \funtwo_{s_1}\ar@{=>}[r]\ar@{=>}@/^1.5pc/[rrr]\ar@{=>}@/^2pc/[rrrr]\ar@{.>}[rrd] & 
      \funtwo_{s_2}\ar@{=>}@/^.7pc/[rr]\ar@{=>}@/^1.5pc/[rrr]\ar@{=>}[rd] & \cdots & 
      \funtwo_{s_{x-1}}\ar@{=>}[r]\ar@{=>}[ld] & \funtwo_{s_x}\ar@{=>}[lld] \\
      & & & \tgfour & & \\
    }
  \end{displaymath}
  where:
  \begin{varitemize}
  \item
    $\conone_{s_1},\ldots,\conone_{s_x}$ are the vertices of $\tgthree$ reachable from
    $\funone$ by following its leftmost outgoing arc, ordered topologically;
  \item
    $x\leq\psize{\tgthree}{j}\leq\psize{\tgone}{j}$
  \item
    $\psize{\tgfour}{t}\leq\psize{\tgthree}{t}\leq\psize{\tgone}{t}$ for every $t$.
  \end{varitemize}
  In $m_x+m_{x-1}+\ldots+m_2$ rewriting steps, the graph becomes
  
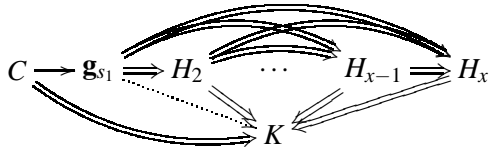
\begin{displaymath}
    \xymatrix@R=10pt@C=15pt{
      & & & & & \\
      \contone \ar[r]\ar@{=>}@/_1pc/[rrrd] & \funtwo_{s_1}\ar@{=>}[r]\ar@{=>}@/^1.5pc/[rrr]\ar@{=>}@/^2pc/[rrrr]\ar@{.>}[rrd] & 
      \tgtwo_{2}\ar@{=>}@/^.7pc/[rr]\ar@{=>}@/^1.5pc/[rrr]\ar@{=>}[rd] & \cdots & 
      \tgtwo_{x-1}\ar@{=>}[r]\ar@{=>}[ld] & \tgtwo_x\ar@{=>}[lld] \\
      & & & \tgfour & & \\
    }
  \end{displaymath}
  Finally, in $m_1$ rewriting steps, we get to
  
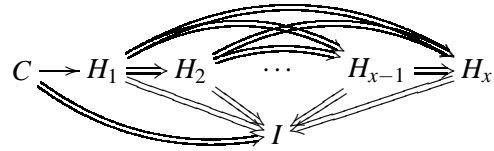
\begin{displaymath}
    \xymatrix@R=10pt@C=15pt{
      & & & & & \\
      \contone \ar[r]\ar@{=>}@/_1pc/[rrrd] & \tgtwo_{1}\ar@{=>}[r]\ar@{=>}@/^1.5pc/[rrr]\ar@{=>}@/^2pc/[rrrr]\ar@{=>}[rrd] & 
      \tgtwo_{2}\ar@{=>}@/^.7pc/[rr]\ar@{=>}@/^1.5pc/[rrr]\ar@{=>}[rd] & \cdots & 
      \tgtwo_{x-1}\ar@{=>}[r]\ar@{=>}[ld] & \tgtwo_x\ar@{=>}[lld] \\
      & & & \tgfive & & \\
    }
  \end{displaymath}
  where $\psize{\tgfive}{t}\leq\psize{\tgfour}{t}$. 
  The graph above can be written as follows
  \begin{displaymath}
    \xymatrix@R=10pt@C=15pt{
       \contone \ar[r]\ar@/_1pc/@{=>}[r] & \tgsix \\
     }
  \end{displaymath}
  Reasoning exactly as in the previous inductive case, we can get
  the following bounds for every $1\leq s\leq x$:
  \begin{align*}
    \psize{\tgtwo_s}{i}&\leq p_{i_s}(\sum_{t=i+1}^\infty\psize{\tgfour}{t})\leq p_{i_s}(\sum_{t=i+1}^\infty\psize{\tgone}{t})\leq p_{i_s}(\size{\tgone})\\
    \psize{\tgtwo_s}{t}&=0\mbox{ whenever $t\neq i$}\\
    \psize{\tgsix}{i}&=\sum_{s=1}^x\psize{\tgtwo_s}{i}+\psize{\tgfive}{i}
        \leq\sum_{s=1}^x p_{i_s}(\sum_{t=i+1}^\infty\psize{\tgone}{t})+\psize{\tgone}{i}\\
     m_s&\leq p_{i_s}(\size{\tgfour}+1+\sum_{t=s+1}^{x}\size{\tgtwo_t})
        \leq p_{i_s}(\size{\tgone}+1+\sum_{t=s+1}^{x}p_{i_t}(\size{\tgone}))
   \end{align*}
  A bound for the size of the intermediate values produced in the
  any of the $k$ groups of steps can be obtained analogously. The thesis
  follows.
\condinc{
\item
  \begin{sloppypar}
  Suppose  $\tfunone=\cond{\tfuntwo_1,\ldots,\tfuntwo_t}:\terms{\sigone}^{n+1}\rightarrow\terms{\sigone}$ where each $\tfuntwo_i$ 
  is tiered recursive with tiers $(j,\ldots, j, i_1, \ldots, i_n)\rightarrow i$. Let  $p_{\tfuntwo_1},\ldots,p_{\tfuntwo_t}$ 
  be some polynomials satisfying the properties above, whose existence follows from the inductive hypothesis.
  Given the initial graph
  \end{sloppypar}
  \begin{displaymath}
  \xymatrix@R=15pt@C=15pt{
    \contone \ar[r]\ar@/_1pc/@{=>}[rr] & \funone\ar[r]\ar@/_1pc/@{=>}[rr] & \conone_s\ar@{=>}[r] & \tgthree \\
  }
  \end{displaymath}
  if we apply the rules $\rho_{s}$, we obtain a graph
  \begin{displaymath}
  \xymatrix@R=15pt@C=15pt{
    \contone \ar@{=>}[r]\ar@/_1pc/[rr] & \conone_s\ar@/_1pc/@{=>}[rr] & \funtwo_s\ar@{=>}[r] & \tgthree \\
  }
  \end{displaymath}
  and the induction hypothesis applied to $\tfuntwo_s$ easily leads to the thesis.}{}
\end{varitemize}
This concludes the proof.
\end{proof}
Theorem~\ref{theo:mainresult} follows easily from Proposition~\ref{prop:inductionload}, once
we observe that any $\sgrone_\tfunone $ is polytime presentable. Indeed, even
if $\sgrone_\tfunone$ is infinite whenever $\tfunone$ is defined by tiered recursion, the rules
in $\sgrone_\tfunone$ are very ``regular'' and an algorithm $\mathcal{A}_{\sgrone_\tfunone}$ with
the required properties can be defined naturally.

\section{Conclusions}\label{Sec:Conclusion}

We proved that Leivant's characterization of polynomial time functions  holds for any free algebra (the original 
result was proved only for algebras with unary constructors, i.e.\  for {word algebras}).
The representation of the terms of the algebras as term graphs permits to avoid 
uncontrolled duplication of shared subterms, thus preserving polynomial bounds.

The main contribution of the paper is the implementation of tiered recursion via term graph 
rewriting. The proofs of the related theorems and propositions  are non-trivial. We 
introduce graph unfolding rules and graph contexts, in order to implement recursion efficiently
and to prove inductively our main result. 
Moreover, the result is given on infinite graph rewriting systems --
in presence of an infinite set of rewriting rules, some well-known computability 
results are lost.


\condinc{\bibliographystyle{plain}}{\bibliographystyle{plain}}
\bibliography{biblio}
\end{document}